\documentclass[envcountsame]{svmult}

\usepackage{graphicx}        
\usepackage{multicol}        

\usepackage{amssymb,amsfonts,amsmath,amsbsy}

\usepackage{rotate,ifthen,latexsym,umlaut}

\usepackage[notref,notcite,final]{showkeys}

\begin{document}

\frontmatter




\mainmatter




\def\epsi{\varepsilon}
\def\D{\mathrm{d}}
\def\I{\mathrm{i}}
\def\E{\mathrm{e}}

\def\re{\mathrm{Re}}
\def\tr{\mathrm{tr}}
\def\Tr{\mathrm{Tr}}

\def\p{{\bf P}}             
\def\T{\mathbb{T}}
\def\A{\mathcal{A}}

\def\N{\mathbb{N}}
\def\R{\mathbb{R}}
\def\C{\mathbb{C}}

\renewcommand{\theequation}{\thesection.\arabic{equation}}


\def\1{\mathbf{1}}        
\def\s{\eta}            
\def\Hf{\mathcal{H}_{\mathrm{f}}}      

\def\Hi{\mathcal{H}}
\def\B{\mathcal{B}}
\def\U{\mathcal{U}}
\def\F{\mathcal{F}}
\def\Or{\mathcal{O}}
\def\ph{\varphi}

\def\({\left(}
\def\){\right)}


\title*{Motion of electrons in adiabatically perturbed
periodic structures}
\titlerunning{Periodic structures}

\toctitle{Motion of electrons in adiabatically\protect\newline
  perturbed periodic structures}
\author{Gianluca Panati\inst{1}\and Herbert Spohn\inst{1} \and
Stefan Teufel\inst{2}}
\authorrunning{G. Panati, H. Spohn, and S. Teufel}
\tocauthor{Gianluca Panati, Herbert Spohn, and Stefan Teufel }
\institute{Zentrum Mathematik, TU M\"unchen, Boltzmannstr. 3,
 85747 Garching \and
Mathematisches Institut, Universit\"{a}t T\"{u}bingen, Auf der Morgenstelle 10,
 72076 T\"{u}bingen\\[3mm]
\texttt{panati@ma.tum.de}, \texttt{spohn@ma.tum.de},
\texttt{stefan.teufel@uni-tuebingen.de}}
%
%
\maketitle \numberwithin{theorem}{section} \numberwithin{equation}{section}
\numberwithin{figure}{section}


\begin{abstract}
We study the motion of electrons in a periodic background potential (usually resulting
from a crystalline solid). For small velocities one would use either the non-magnetic or the magnetic Bloch hamiltonian, while in the relativistic regime one would use the Dirac equation with a periodic potential. The dynamics, with the background potential included, is perturbed either through slowly varying external electromagnetic potentials or through a slow deformation of the crystal. In either case we discuss how the Hilbert space of states decouples into almost invariant subspaces and explain the effective dynamics within such a subspace.
\end{abstract}


\section{Introduction}\label{PST_sec.1}
\setcounter{equation}{0}

In a crystalline solid the conduction electrons move in the potential created by the
ions and the core electrons. Somewhat mysteriously and linked to the Pauli exclusion
principle, the Coulomb repulsion between conduction electrons may be ignored, within
a good approximation. Thereby one arrives at a fundamental model of solid state
physics, namely an ideal Fermi gas of electrons subject to a periodic crystal
potential. Let $\Gamma$ be the periodicity lattice. It is a Bravais lattice and
generated through the basis $\{\gamma_1,\gamma_2,\gamma_3\}$,
$\gamma_j\in\mathbb{R}^3$, as
\begin{equation}\label{PST_1.1}
\Gamma=\{\gamma=\sum^3_{j=1}\alpha_j\gamma_j\quad \textrm{with
}\alpha\in\mathbb{Z}^3\}\,.
\end{equation}
The crystal potential $V_\Gamma$ is then $\Gamma$-periodic, i.e.,
$V_\Gamma:\mathbb{R}^3\to\mathbb{R}$ and
$V_\Gamma(x+\gamma)=V_\Gamma(x)$ for all $\gamma\in\Gamma$, and the
electrons are governed by the one-particle hamiltonian
\begin{equation}\label{PST_1.2}
H_{\mathrm{SB}}=-\frac{1}{2}\Delta_x+V_\Gamma\,.
\end{equation}
$H_{\mathrm{SB}}$ is the (Schr\"{o}dinger)--Bloch hamiltonian. A wave function
$\psi_t\in L^2(\mathbb{R}^3)$ evolves in time according to the Schr\"{o}dinger
equation
\begin{equation}\label{PST_1.3}
i\frac{\partial}{\partial t}\psi_t=H_{\mathrm{SB}}\psi_t\,.
\end{equation}
We have chosen units such that the mass of an electron
$m_\mathrm{e}=1$ and $\hbar=1$. The electron charge, $e$, is
absorbed in $V_\Gamma$. Since $V_\Gamma$ is periodic, electrons move
ballistically with an effective dispersion relation given by the
Bloch energy bands $E_n$, see below for a precise definition. $E_n$
is periodic with respect to the lattice $\Gamma^\ast$ dual to
$\Gamma$, $E_n(k+\gamma^\ast)=E_n(k)$ for all
$\gamma^\ast\in\Gamma^\ast$, $k\in\mathbb{R}^3$. This feature makes
the dynamical properties of a Bloch electron very different from a
massive particle with dispersion
$E_{\textrm{free}}(k)=\frac{1}{2}k^2$ valid in case $V_\Gamma=0$.

The thermodynamics of the electron gas is studied taking
$H_\mathrm{SB}$ as a starting point. Dynamically, however, one wants
to probe the response of the electrons to external forces which very
crudely come
in two varieties.\medskip\\
(i) \textit{External electromagnetic potentials}. Electrostatic
potentials manufactured in a lab have a slow variation on the scale
of the lattice $\Gamma$. Therefore we set
$V_\mathrm{ext}(x)=e\phi(\varepsilon x)$, $e$ the charge of the
electron, with $\varepsilon$ a dimensionless parameter and $\phi$
independent of $\varepsilon$. $\varepsilon\ll 1$ means that the
potential $V_\mathrm{ext}$ has a slow variation when measured with
respect to the lattice spacing of $\Gamma$. Note that the
electrostatic force is $\mathcal{O}(\varepsilon)$ and thus weak.
External magnetic fields on the other hand can be so strong that the
radius of gyration is comparable to the lattice spacing. It then
makes sense to split the vector potential as $A_0+A_\mathrm{ext}$,
where $A_0(x)=-\frac{1}{2}B_0 \wedge x$ with
$B_0\in\mathbb{R}^3$ a constant magnetic field. Included in
$H_\mathrm{SB}$, this yields the magnetic Bloch hamiltonian
\begin{equation}\label{PST_1.4}
H_\mathrm{MB}=\frac{1}{2}(-i\nabla_x- A_0)^2+V_\Gamma\,.
\end{equation}
$A_\mathrm{ext}$ is a probing vector potential in addition to $A_0$.
$A_\mathrm{ext}$ is slowly varying on the scale of the lattice,
$A_\mathrm{ext}(x)=A(\varepsilon x)$ with $A$ independent of
$\varepsilon$, and the corresponding magnetic field is small of
order $\varepsilon$. Including all electromagnetic potentials,
for simplicity with the electric charge absorbed into $A$ and $\phi$, the
hamiltonian becomes
\begin{equation}\label{PST_1.5}
H=\frac{1}{2}\big(-i\nabla_x-A_0(x)-A(\varepsilon x)\big)^2
+V_\Gamma(x)+\phi(\varepsilon x)\,.
\medskip
\end{equation}
(ii) \textit{Mechanical forces}. The crystal lattice can be deformed through
external pressure and shear. Thereby an electric polarization is induced, an effect
which is known as piezoelectricity. If charges are allowed to flow, in this way
mechanical pressure can be transformed into electric currents. The mechanical forces
are time-dependent but slow on the typical time-scale of the electrons. Therefore in
(\ref{PST_1.2}) $V_\Gamma(x)$ is replaced by $V_{\Gamma(\varepsilon
t)}(x,\varepsilon t)$. $\Gamma(\varepsilon t)$ is the instantaneous periodicity
lattice and is defined as in (\ref{PST_1.1}). $V_{\Gamma(t)}$ is space-periodic,
i.e.~$V_{\Gamma(t)}(x+\gamma,t)=V(x,t)$ for all $\gamma\in\Gamma(t)$. The special
case of a time-independent lattice, $\Gamma(t)=\Gamma$, but a still slowly in time
varying crystal potential is also of interest. For example, one may imagine a unit
cell with two nuclei. If the two nuclei are displaced relative to each other, then
$\Gamma$ remains fixed while the crystal potential in the unit cell changes with
time. The resulting piezoelectric hamiltonian reads
\begin{equation}\label{PST_1.6}
H_\mathrm{PE}(t)=-\frac{1}{2}\Delta_x+V_{\Gamma(\varepsilon
t)}(x,\varepsilon t)\,.
\end{equation}\medskip

Our general goal is to understand, in each case, the structure of the solution of
the time-dependent Schr\"{o}dinger equation for small $\varepsilon$. Obviously, $H$
in (\ref{PST_1.5}) is a space-adiabatic problem, while (\ref{PST_1.6}) corresponds
to a time-adiabatic problem. However in the latter case it turns out to be
profitable to transform to a time-independent lattice, say $\Gamma(0)$. Then also
terms varying slowly in space are generated. Thus, in the general case the full
power of the space-adiabatic perturbation theory \cite{1,2} will be needed. A word
of caution must be issued here for the magnetic Bloch hamiltonian. To use the
methods from \cite{1} in this context, the unperturbed Hamiltonian must be periodic,
which is the case only if the magnetic flux per unit cell is rational. One can then
define an enlarged magnetic unit cell such that $H_\mathrm{MB}$ is invariant with
respect to magnetic translations. If the magnetic flux is not rational, the crutch
is to include in $A_0$ a nearby rational flux part of the magnetic field, with a
small denominator, and to treat the remainder as $A_\mathrm{ext}$.

To achieve our goal, depending on the context we use one of the periodic
hamiltonians as backbone. The periodic hamiltonian is denoted by $H_\mathrm{per}$
with  $H_\mathrm{per}$ either $H_\mathrm{SB}$, or  $H_\mathrm{MB}$, or
$H_\mathrm{PE}$ at fixed $t$, or $H_\mathrm{LS}$ from (\ref{PST_1.8}), or
$H_\mathrm{DB}$ from (\ref{PST_1.9}). As explained below, the Hilbert space
$\mathcal{H}=L^2(\mathbb{R}^3)$ then splits as
$\mathcal{H}=\bigoplus^\infty_{n=0}\mathcal{H}_n$, where $n$ is the band index. Each
subspace $\mathcal{H}_n$ is invariant under $\exp[-itH_\mathrm{per}]$ and
$H_\mathrm{per}$ restricted to $\mathcal{H}_n$ is unitarily equivalent to
multiplication by $E_n(k)$. $E_n(k)$ is the effective hamiltonian associated to the
$n$-th band. The complexity of the full problem has been reduced substantially,
since only a single band dynamics has still to be studied. Modifying
$H_\mathrm{per}$ such that it becomes slowly varying in space-time is, vaguely
speaking, a small perturbation. Thus one would expect that the invariant subspace
$\mathcal{H}_n$ is to be substituted by a slightly tilted subspace. On this subspace
$E_n(k)$ will turn into a more complicated effective hamiltonian. The difficulty is
that, while the dynamics generated by $H_\mathrm{per}$ can be computed by solving a
purely spectral problem, none of the perturbed hamiltonians can be understood in
this way. In particular, one has to spell out carefully over which time scale the
slightly tilted subspace associated to $\mathcal{H}_n$ remains approximately
invariant and in what sense the dynamics generated by the effective hamiltonian
approximates the true time evolution.

To lowest order the effective hamiltonian can be guessed from
elementary considerations and belongs to a standard tool of solid
state physics \cite{3}. The guess provides however little hint on
the validity of the approximation. There one needs a mathematical
theorem which states precise conditions on the initial wave function
and provides an error bound, from which the time scale for validity
can be read off.

Under the header ``geometric phase'' physicists and quantum chemists
have realized over the past twenty years, say, that the first order
correction to the effective hamiltonian carries a lot of interesting
physics, see \cite{4} for a recent comprehensive overview. For the
magnetic Bloch hamiltonian the first order correction yields a Hall
current proportional to the Chern number of the magnetic Bloch
vector bundle. Similarly, the modern theory of piezoelectricity,
expresses the piezocurrent as an integral of the Berry connection
over the Brillouin zone, see King-Smith, Vanderbilt \cite{5} and
Resta \cite{6}. First order effective Hamiltonians are no longer
guessed so easily and it is convenient to have the systematic scheme
\cite{1} available.

In nature electrons are spin $\frac{1}{2}$ particles. The wave function is thus
$\mathbb{C}^2$-valued and the hamiltonian in (\ref{PST_1.5}) is modified to
\begin{equation}\label{PST_1.7}
H=\frac{1}{2}\big(-i\nabla_x-A_0(x)-A(\varepsilon x)\big)^2
+V_\Gamma(x)+\phi(\varepsilon
x)-\frac{1}{2}\sigma\cdot\big(B_0+\varepsilon B(\varepsilon x)\big)
\end{equation}
with $B=\nabla\wedge A$ for the slowly varying part of the magnetic
field. Here $\sigma=(\sigma_1,\sigma_2,\sigma_3)$ is the 3-vector of
Pauli spin matrices. Besides the term proportional to the uniform
magnetic field $B_0$, $H$ acquires a subleading term of order
$\varepsilon$. More accurately one may want to include the
spin-orbit coupling. The periodic piece of the hamiltonian reads
then
\begin{equation}\label{PST_1.8}
    H_\mathrm{LS}=-\frac{1}{2}\Delta_x+V_\Gamma(x)+\frac{1}{4}\sigma
    \cdot\big(\nabla V_\Gamma(x)\wedge(-i\nabla_x)\big)
\end{equation}
and the slowly varying potential is added as in (\ref{PST_1.7}) with the additional
subleading term $\varepsilon\frac{1}{4}\sigma\cdot (\nabla \phi(\varepsilon
x)\wedge(-i\nabla_x))$.

Depending on the crystalline solid, the conduction electrons can
move so fast that relativistic corrections become important. On the
one-particle level an obvious choice is then the Dirac equation with
a periodic potential $V_\Gamma$. Wave functions are
$\mathbb{C}^4$-valued and the hamiltonian reads
\begin{equation}\label{PST_1.9}
H_\mathrm{DB}=\beta m_\mathrm{e}c^2+c\alpha\cdot p+V_\Gamma\,,\quad
p=-i\nabla_x\,.
\end{equation}
We introduced here the mass, $m_\mathrm{e}$, of the electron and the
speed of light, $c$. The $4\times 4$ matrices
$\beta,\alpha_1,\alpha_2,\alpha_3$ are standard and defined in
\cite{7,8}, for example. Note that the Lorentz frame is fixed by the
solid, i.e.~by $V_\Gamma$.

In fact, the non-relativistic limit for $H_\mathrm{DB}$ yields the spin-orbit
hamiltonian $H_\mathrm{LS}$ \cite{7,8}. If $\|V_\Gamma\|$ is bounded, for
sufficiently large $c$, the Dirac hamiltonian $H_\mathrm{DB}$ has a spectral gap,
which widens as $c\to\infty$. Projecting onto the electron subspace, to leading
order in $1/c$ one obtains the Pauli-Bloch hamiltonian $-(1/2
m_\mathrm{e})\Delta_x+V_\Gamma$ with the spin-orbit coupling in (\ref{PST_1.8}) as a
correction of strength $1/(m_\mathrm{e}c)^2$. In addition the crystal potential is
corrected by $-\Delta_xV_\Gamma(x)/8(m_\mathrm{e}c)^2$.

In our contribution we will provide some background on how to establish, including
error bounds, the validity of the approximate dynamics as generated by an effective
hamiltonian, including order $\varepsilon$ corrections, for most of the models
mentioned in the introduction. For this purpose it is necessary to briefly recall
the spectral theory for the periodic hamiltonian, which is done in the following
section. In the subsequent sections we deal with particular cases in more detail. We
start with the non-magnetic Bloch hamiltonian, see (\ref{PST_1.5}) with $B_0=0$. For
the magnetic Bloch hamiltonian we explain how $B_0\to 0$ and $B_0\to\infty$ may be
viewed as particular adiabatic limits.
Piezoelectricity is discussed in the last section.\smallskip\\
\textbf{Remark}. Our contribution is one part of the research
project jointly with S.~Bauer and M.~Kunze within the Schwerpunkt.
Their part will be covered in \cite{9}. Both contributions appear
now as almost disjoint, which only reflects that we wanted to
present a coherent story. The unifying aspect is an adiabatic limit
for wave-type evolution equations. In this contribution we stay on
the level of effective hamiltonians while in \cite{9} one pushes the
scheme to the first dissipative correction.

\section{The periodic hamiltonians}\label{PST_sec.2}
\setcounter{equation}{0}

We consider a general dimension, $d$, and assume that the
periodicity lattice $\Gamma$ is represented as
\begin{equation}\label{PST_2.1a}
 \Gamma =\Big\{ x\in\mathbb{R}^d: x=
\textstyle{\sum_{j=1}^d}\alpha_j\,\gamma_j \,\,\,\mbox{for
some}\,\,\alpha \in \mathbb{Z}^d \Big\}\,,
\end{equation}
 where
$\{\gamma_1,\ldots,\gamma_d \}$ are vectors spanning $\mathbb{R}^d$.
We denote by $\Gamma^*$ the dual lattice of $\Gamma$ with respect to
the standard inner product in $\mathbb{R}^d$, i.e.~the lattice
generated by the dual basis $\{\gamma_1^*,\ldots,\gamma_d^*\}$
determined through the conditions $\gamma_i^* \cdot \gamma_j = 2\pi
\delta_{ij}$, $i,j\in\{1,\ldots,d\}$. The centered fundamental
domain $M$ of $\Gamma$ is defined by
\begin{equation}\label{PST_2.2}
 M = \Big\{ x\in\mathbb{R}^d: x=
\textstyle{\sum_{j=1}^d}\beta_j\,\gamma_j
\,\,\,\mbox{for}\,\,\beta_j\in
[-\textstyle{\frac{1}{2},\frac{1}{2}}]
 \Big\}\,,
\end{equation}
 and analogously the centered fundamental domain $M^\ast$ of $\Gamma^\ast$.  The set $M^\ast$ is
 the {\em first Brillouin zone} in the physics parlance.

\medskip

\noindent \textbf{Assumption 1.}\ \textit{The crystal potential
 $V_\Gamma:\mathbb{R}^d\to\mathbb{R}$ satisfies $V_\Gamma(x+\gamma)=V_\Gamma(x)$
 for all $\gamma\in\Gamma$, $x\in\mathbb{R}^d$. $V_\Gamma$ is infinitesimally
 bounded with respect to $-\Delta $.}

\medskip

 It follows from Assumption~1 that the periodic hamiltonians discussed below are self-adjoint
 on the domain of $-\Delta$.

\subsection{The Bloch hamiltonian}\label{PST_sec.2a}

We consider
\begin{equation}\label{PST_2a.1a}
H=-\frac{1}{2}\Delta+V_\Gamma\,.
\end{equation}
The periodicity of $H$ is exploited through the Bloch-Floquet-Zak
transform, or just the Zak transform for sake of brevity \cite{10}.
The advantage of such a variant is that the fiber at $k$ of the
transformed Hamiltonian operator has a domain which does not depend
on $k$.

The Zak transform is defined as
\begin{equation}\label{PST_2a.1}
(\mathcal{U}_\mathrm{Z}\psi)(k,x):=\sum_{\gamma\in\Gamma}
\mathrm{e}^{-\mathrm{i}k \cdot (x+\gamma)} \psi(x+\gamma)\,, \qquad
(k,x)\in\mathbb{R}^{2d},
\end{equation}
initially for a fast-decreasing function $\psi\in\mathcal{S}(\mathbb{R}^d)$. One
directly reads off from (\ref{PST_2a.1}) the following periodicity properties
\begin{equation}\label{PST_2a.2}
\big(\mathcal{U}_\mathrm{Z}\psi\big) (k, y+\gamma) = \big(
\mathcal{U}_\mathrm{Z}\psi\big) (k,y)\quad \mbox{ for all} \quad
\gamma\in\Gamma\,,
\end{equation}
\begin{equation}\label{PST_2a.3}
\big(\mathcal{U}_\mathrm{Z}\psi\big) (k+\lambda, y) =
\mathrm{e}^{-\mathrm{i}y\cdot\lambda}\,\big(\mathcal{U}_\mathrm{Z}\psi\big)
(k,y) \quad\mbox{ for all} \quad \lambda\in\Gamma^*\,.
\end{equation}
From (\ref{PST_2a.2}) it follows that, for any fixed $k\in{\mathbb{R}^d}$,
$\big(\mathcal{U}_\mathrm{Z}\psi \big)(k,\cdot)$ is a $\Gamma$-periodic function and
can thus be regarded as an element of $\mathcal{H}_\mathrm{f} = L^2(M)$.
$M=\mathbb{R}^d/\Gamma$ and it has the topology of the $d$-dimensional torus
$\mathbb{T}^d$. On the other side, Equation (\ref{PST_2a.3}) involves a unitary
representation of the group of lattice translations on $\Gamma^*$ (isomorphic to
$\Gamma^*$ and denoted as $\Lambda$), given by
\begin{equation} \label{PST_2a.4}
\tau:\Lambda \to \mathcal{U}(\mathcal{H}_\mathrm{f})\,,\quad\lambda
\mapsto \tau(\lambda)\,, \quad \big(\tau(\lambda)\varphi \big)(y) =
\mathrm{e}^{\mathrm{i}y\cdot\lambda}\varphi(y).
\end{equation}
It is then convenient to introduce the Hilbert space
\begin{eqnarray}\label{PST_2a.5}
&&\hspace{-40pt}\mathcal{H}_\tau =\big\{ \psi\in L^2_{\rm
loc}(\mathbb{R}^d, \mathcal{H}_\mathrm{f}):\,\, \psi(k - \lambda) =
\tau(\lambda)\,\psi(k) \qquad \textrm{for all } \lambda \in \Lambda
\big\}\nonumber\\
&&\hspace{-18pt}= L^2_\tau(\mathbb{R}^d\,,\;\mathcal{H}_\mathrm{f})\,,
\end{eqnarray}
equipped with the inner product
\begin{equation}\label{PST_2a.6}
\langle \psi,\,\varphi\rangle_{\mathcal{H}_\tau} = \int_{M^{*}} dk\,
\langle \psi(k),\,\varphi(k)\rangle_{\mathcal{H}_\mathrm{f}}\,.
\end{equation}
Obviously, there is a natural isomorphism  between $\mathcal{H}_\tau$ and
$L^2(M^{*}, \mathcal{H}_\mathrm{f})$ given by restriction from $\mathbb{R}^d$ to
$M^{*}$, and with inverse given by $\tau$-equivariant continuation, as suggested by
(\ref{PST_2a.3}). Equipped with these definitions, one checks that the map in
(\ref{PST_2a.1}) extends to a unitary operator
\begin{equation}\label{PST_2a.7}
\mathcal{U}_\mathrm{Z}: L^2(\mathbb{R}^d)\to \mathcal{H}_\tau \cong
L^2(M^{*}, L^2 (M)),
\end{equation}
with inverse given by
\begin{equation}\label{PST_2a.8}
(\mathcal{U}_\mathrm{Z}^{-1}\varphi)(x) = \int_{M^*} dk \,
\mathrm{e}^{\mathrm{i} k \cdot x} \varphi(k, [x]),
\end{equation}
where $[\, \cdot \, ]$ refers to  the a.e. unique decomposition $x =
\gamma_x + [x]$, with $\gamma_x \in \Gamma$ and $[x] \in M$.

As already mentioned, the advantage of this construction is that the transformed
hamiltonian is a fibered operator over $M^*$. Indeed, for the Zak transform of the
hamiltonian operator (\ref{PST_2a.1a}) one finds
\begin{equation}\label{PST_2a.9}
\mathcal{U}_\mathrm{Z}H \mathcal{U}_\mathrm{Z}^{-1} =
\int_{M^{*}}^\oplus dk\,H(k)
\end{equation}
with fiber operator
\begin{equation}\label{PST_2a.10}
H(k) = \frac{1}{2}\big( -{\mathrm{i}} \nabla_y + k\big)^2 +
V_\Gamma(y)\,,  \quad k\in M^{*} \,.
\end{equation}
By Assumption 1, for fixed $k\in M^{*}$, the operator $H(k)$ acts on $L^2(M)$ with
the Sobolev space $H^2(M)$ as domain independently of $k\in M^{*}$. Each fiber
operator $H(k)$ has pure point spectrum accumulating at infinity. For definiteness
the eigenvalues are enumerated according to their magnitude $E_0(k) \leq E_1(k) \leq
E_2(k) \leq \ldots$ and repeated according to their multiplicity. $E_n:M^\ast\to
\mathbb{R}$ is the $n$-th energy band function. It is continuous on $M^\ast$ when
viewed as a $d$-torus. Generically the eigenvalues $E_n(k)$ are non-degenerate. Of
course, there may be particular points in $k$-space where particular energy bands
touch each other and the corresponding eigenvalue becomes degenerate. The normalized
eigenfunction corresponding to $E_n(k)$ is the Bloch function and denoted by
$\varphi_n(k)\in H^2(M)$. It is determined only up to a $k$-dependent phase factor.
A further arbitrariness comes from points where energy bands touch. We denote by
$P_n(k)$ the projection along $\varphi_n(k)$ and set
\begin{equation}\label{PST_2a.11}
P_n=\int^\oplus_{M^\ast}dk P(k)\,,\qquad
\mathcal{H}_n=P_nL^2(\mathbb{R}^d)\,.
\end{equation}

Through the Zak transform we have achieved the product structure
\begin{equation}\label{PST_2a.12}
\mathcal{H}=\mathcal{H}_\mathrm{s}\otimes\mathcal{H}_\mathrm{f}\,,\qquad
\mathcal{H}_\mathrm{s}=L^2(M^\ast)\,,\;\mathcal{H}_\mathrm{f}=L^2(M)\,.
\end{equation}
$\psi\in\mathcal{H}_\mathrm{n}$ is of the form
$\phi(k)\varphi_n(k,y)$. The band index $n$ fixes the local pattern
of the wave function $\psi$ while $\phi(k)$ provides the slow
variation. Therefore $L^2(M^\ast)$ is the Hilbert space of the slow
degrees of freedom. On the other hand for fixed $k$, one has
oscillations in time determined by the eigenvalues $E_n(k)$. On long
time scales, these become fast oscillations. Therefore
$\mathcal{H}_\mathrm{f}=L^2(M)$ is the Hilbert space of the fast
degrees of freedom.

Since $[P_n,H]=0$, the subspaces $\mathcal{H}_n$ are invariant under
$\mathrm{e}^{-\mathrm{i}Ht}$. $P_n\mathrm{e}^{-\mathrm{i}Ht}P_n$ is
unitarily equivalent to multiplication by
$\mathrm{e}^{-\mathrm{i}E_n(k)t}$ on $L^2(M^\ast)$. Note that, in
general, $\mathcal{H}_n$ is not a spectral subspace for $H$. The
band functions generically have overlapping ranges. Therefore, if
slowly varying terms are added to the hamiltonian, the dynamics can
no longer be captured so easily by a spectral analysis of the
perturbed hamiltonian.

\subsection{The magnetic Bloch hamiltonian}\label{PST_sec.2b}

We consider $d=3$. The hamiltonian reads
\begin{equation}\label{PST_2b.1}
H=\frac{1}{2}\big(-i\nabla_x-A(x)\big)^2+V_\Gamma(x)\,,\quad
x\in\mathbb{R}^3\,,
\end{equation}
with $A(x)=-\frac{1}{2}B\wedge x$, $B\in\mathbb{R}^3$. Physically
the most relevant case is $d=2$. It is included here by setting
$x=(x_1,x_2,0)$ and $B=(0,0,B_0)$. Following Zak \cite{11}, see also
\cite{12}, one introduces the magnetic translations
\begin{equation} \label{PST_Magnetic translations}
(T_\alpha\psi)(x)=\big(e^{-i\alpha\cdot(-i\nabla_x+A(x))}\psi\big)(x)=
e^{i\alpha A(x)}\psi(x-\alpha)
\end{equation}
with $\alpha\in\mathbb{R}^3$. They satisfy the Weyl relations
\begin{equation}\label{PST_2b.3}
T_\alpha T_\beta=e^{-\frac{i}{2}B\cdot(\alpha\wedge\beta)}
T_{\alpha+\beta}=e^{-i B\cdot(\alpha\wedge\beta)}T_\beta T_\alpha\,.
\end{equation}
To have a commuting subfamily we need

\medskip

\noindent \textbf{Assumption 2.}\ \textit{The magnetic field $B$ is such that
$B\cdot(\gamma\wedge\gamma')\in 2\pi\mathbb{Q}$ for all $\gamma,\gamma'\in\Gamma$.}

\medskip

In the two-dimensional case our assumption requires that the
magnetic flux per unit cell, $B_0\cdot(\gamma_1\wedge\gamma_2)$, is
a rational multiple of $2\pi$.

Under the Assumption 2 there exists a sublattice $\Gamma_0\subset\Gamma$ such that
$B\cdot(\gamma\wedge\gamma')\in 2\pi\mathbb{Z}$ for every
$\gamma,\gamma'\in\Gamma_0$. $\Gamma_0$ is not unique. The set
$\{T_\alpha\}_{\alpha\in\Gamma_0}$ is a family of commuting operators, which commute
with $H$. Since $T_\alpha T_\beta=\pm T_{\alpha+\beta}$, the magnetic translations
still form only a projective group. It becomes a group by an even smaller sublattice
$\Gamma_1\subset \Gamma_0$ such that $B\cdot(\gamma\wedge\gamma')\in 4\pi\mathbb{Z}$
for all $\gamma,\gamma'\in\Gamma_1$. Another common choice is a further modification
of the phase through
\begin{equation}\label{PST_2b.3bis}
\mathcal{T}_\alpha=e^{-\frac{i}{2}\varphi(\alpha)}T_\alpha
\end{equation}
with
$\varphi(\alpha)=B_1\alpha_2\alpha_3+B_3\alpha_1\alpha_2-B_2\alpha_1\alpha_3$.
Then
$\mathcal{T}_\alpha\mathcal{T}_\beta=\mathcal{T}_{\alpha+\beta}$ for
all $\alpha,\beta\in\Gamma_0$.

We can now proceed as in the non-magnetic case. The Zak transform
becomes
\begin{equation}\label{PST_2b.4}
(\mathcal{U}_\mathrm{Z}\psi)(k,x)=\sum_{\gamma\in\Gamma_0}e^{-ik\cdot(x+\gamma)}
\mathcal{T}_\gamma\psi(x)\,,\quad (k,x)\in\mathbb{R}^6\,.
\end{equation}
The properties of $\mathcal{U}_\mathrm{Z}\psi$ are as in (\ref{PST_2.2}),
(\ref{PST_2a.1a}) provided $\Gamma$ is replaced by $\Gamma_0$, and
$\mathcal{H}_\tau$ is replaced by $\mathcal{H}^B_\tau=\{u\in
L^2_\mathrm{loc}(\mathbb{R}^d,\mathcal{H}_\mathrm{f}): $ (\ref{PST_2b.7}) below
holds true$\}$. In particular, $H$ of (\ref{PST_2b.1}) admits the fiber
decomposition
\begin{equation}\label{PST_2b.5}
\mathcal{U}_\mathrm{Z}H\mathcal{U}_\mathrm{Z}^{-1}=\int^\oplus_{M^\ast}dk
H(k)
\end{equation}
with $M^\ast$ the first Brillouin zone of $\Gamma^\ast_0$ and with the
fiber operator
\begin{equation}\label{PST_2b.6}
H(k)=\frac{1}{2}(-i\nabla_y+\frac{1}{2}B\wedge y+k)^2+V_\Gamma(y)\,.
\end{equation}
The domain of $H(k)$ is independent of $k$ but, in contrast to $H(k)$ from
(\ref{PST_2a.10}), a function $u$ in the domain has to satisfy the more complicated
boundary condition
\begin{equation}\label{PST_2b.7}
e^{-\frac{i}{2} y\cdot(\alpha\wedge B)}u(y-\alpha)=u(y)\,.
\end{equation}

\subsection{Dirac hamiltonian, spin-orbit coupling}\label{PST_sec.2c}

The Dirac hamiltonian with periodic potential reads
\begin{equation}\label{PST_2c.1}
H=\beta-i \alpha\cdot\nabla_x+V_\Gamma(x)\,,
\end{equation}
where we have set $m_\mathrm{e}=1$, $c=1$. As for the Bloch
hamiltonian, $H$ admits the fiber decompositon
\begin{equation}\label{PST_2c.2}
H=\int^\oplus_{M^\ast} dk H(k)
\end{equation}
with fiber hamiltonian
\begin{equation}\label{PST_2c.3}
H(k)=\beta+\alpha\cdot(-i\nabla_y+k)+V_\Gamma(y)\,.
\end{equation}
$H(k)$ acts on $L^2(M,\mathbb{C}^4)$ with periodic boundary conditions
(\ref{PST_2a.2}). The free Dirac operator has a spectral gap of size 2, in our
units, between the electron and positron subspace. If we assume $\|V_\Gamma\|<1$,
then this gap persists and the eigenvalues can be labelled as $E_0(k)\leq
E_1(k)\leq\ldots$ in the electron subspace and as $E_{-1}(k)\geq
E_{-2}(k)\geq\ldots$ in the positron subspace. One has $E_{-1}(k)<E_0(k)$ for all
$k\in M^\ast$. (In fact the labelling can be achieved without a restriction on
$\|V_\Gamma\|$, see \cite{13}).

For $V_\Gamma=0$, the eigenvalue $E(k)$ is two-fold degenerate
reflecting the spin $\frac{1}{2}$ of the electron, resp.~positron.
This degeneracy persists if the periodic potential is inversion
symmetric, see \cite{13} for details.
\begin{proposition}\label{PST_2c.prop1} Let $H$ be given by
(\ref{PST_2c.1}) with $\|V_\Gamma\|<1$. Let there exist $a\in\mathbb{R}^3$ such that
$V_\Gamma(x+a)=V_\Gamma(-x+a)$. Then each $E_n(k)$ is at least two-fold degenerate.
\end{proposition}
\begin{proof} Without loss of generality we
may assume $a=0$. We use the standard basis for the
$\alpha$-matrices, see \cite{7}. In this basis time-reversal
symmetry is implemented by the anti-unitary operator
\begin{equation}\label{PST_2c.4}
T\psi(y)= -i\alpha_3\alpha_1\psi^\ast(y)\,,
\end{equation}
where the complex conjugation is understood component-wise. Using
that $\alpha_\ell\alpha_3\alpha_1=
-\alpha_3\alpha_1\overline{\alpha}_\ell$, $\ell=1,2,3$, where
$\overline{\alpha}_\ell$ refers to matrix element-wise complex
conjugation, one checks that
\begin{equation}\label{PST_2c.5}
-i\nabla_y\alpha_\ell T=-i T\nabla_y\alpha_\ell\,,\quad k\alpha_\ell
T=-Tk\alpha_\ell
\end{equation}
and therefore
\begin{equation}\label{PST_2c.6}
T^{-1}H(k)T=H(-k)\,.
\end{equation}

Secondly we use space inversion as
\begin{equation}\label{PST_2c.7}
R\psi(y)=\beta\psi(-y)\,.
\end{equation}
Then
\begin{equation}\label{PST_2c.8}
R^{-1}H(k)R=H(-k)\,.
\end{equation}
Combining both symmetries implies
\begin{equation}\label{PST_2c.9}
T^{-1}R^{-1}H(k)RT=H(k)\,.
\end{equation}

If $H(k)\psi=E\psi$, then also $RT\psi$ is an eigenfunction with the
same eigenvalue. Thus our claim follows from
$\langle\psi,RT\psi\rangle=0$. To verify this identity we note that
$-i\alpha_3\alpha_1=\textrm{diag }(\sigma_2,\sigma_2)$ and $\langle
\chi,R\sigma_2\chi^\ast\rangle=0$ for an arbitrary two-spinor
$\chi$.
\end{proof}
\begin{corollary}\label{PST_2c.2a} The eigenvalue $E_n(0)$ of
$H(0)$ is at least two-fold degenerate.
 \end{corollary}
\begin{proof}Since $T^\ast H(0)T=H(0)$ by (\ref{PST_2c.6}) and
$\langle\psi,(-i\alpha_3\alpha_1)\psi^\ast\rangle=0$, the claim
follows.
\end{proof}
 If $V_\Gamma$ is not inversion symmetric, generically an
energy band is two-fold degenerate at $k=0$ and then splits into two
non-degenerate bands. Note that a non-degenerate eigenvalue $E_n(k)$
has an associated eigenvector with a definite spin orientation.

The Pauli equation with spin-orbit coupling has the hamiltonian
\begin{equation}\label{PST_2c.12}
H=-\frac{1}{2}\Delta_x+V_\Gamma(x)+\frac{1}{4}\sigma\cdot\big(\nabla
V_\Gamma(x)\wedge(-i\nabla_x)\big)\,.
\end{equation}
After Zak transform the corresponding fiber hamiltonian becomes
\begin{equation}\label{PST_2c.13}
H(k)=\frac{1}{2}(-i\nabla_y+k)^2+V_\Gamma(y)+
\frac{1}{4}\sigma\cdot\big(\nabla
V_\Gamma(y)\wedge(-i\nabla_y+k)\big)
\end{equation}
with periodic boundary conditions. $H$ of (\ref{PST_2c.12}) is bounded from below.
But otherwise the band structure is similar to the periodic Dirac operator.
Proposition \ref{PST_2c.prop1} and Corollary \ref{PST_2c.2} hold as stated. In the
proof one only has to use the appropriate time-reversal operator, which is
$T\psi=\sigma_2\psi^\ast$ in the $\sigma_3$-eigenbasis.

\subsection{Gap condition and smoothness}\label{PST_sec.2d}

Let us consider one of the periodic hamiltonians, $H_\mathrm{per}$,
with fiber decomposition $H(k)$. $H_\mathrm{per}$ is adiabatically
perturbed to $H^\varepsilon$. Very crudely the corresponding unitary
groups should be close. To make such a notion quantitative a gap
condition must be imposed. We denote by $\sigma(H)$ the spectrum of
the self-adjoint operator $H$.\medskip\\
\textbf{Gap condition:} \textit{We distinguish a family of $m$
physically relevant energy bands $\{E_j(k)\,,\;n\leq j\leq
n+m-1\}=\sigma_0(k)$. This family satisfies the gap condition if}
\begin{equation}\label{PST_2d.1}
\textrm{dist}\{\sigma_0(k)\,,\;\sigma(H(k))\setminus\sigma_0(k)\}\geq
g>0\quad \textrm{for all }k\in M^\ast\,.
\end{equation}
We repeat that the gap condition is not a spectral condition for
$H_\mathrm{per}$. Let us set $P^0=\bigoplus^{n+m-1}_{j=n} P_j$.

Under the gap condition the projector $P(k)$ depends smoothly, in
many cases even (real) analytically, on $k$. $\mathrm{Ran}\,P^0(k)$ is spanned
by the basis $\{\varphi_j(k)\}_{j=n,\ldots,n+m-1}$. If the $m$
relevant energy bands have no crossings amongst each other, then
$\varphi_j$ is necessarily an eigenvector of $H(k)$ satisfying
$H(k)\varphi_j(k)$$=E_j(k)\varphi_j(k)$. But if there are band
crossings, it can be convenient not to insist on $\varphi_j(k)$
being an eigenvector of $H(k)$. Thus, while $P^0(k)$ is unique, the
spanning basis is not. In applications it is of importance to know
whether there is at least some choice of $\varphi_j(k)$,
$j=n,\ldots,n+m-1$, such that they have a smooth $k$-dependence.
Locally, this can be achieved. However, since $M^\ast$ has the
topology of a torus, a global extension might be impossible. In fact
this will generically happen for the magnetic Bloch hamiltonian, see
\cite{14,Novikov,16} for examples. Somewhat surprisingly, a
reasonably general answer has been provided only recently \cite{17}.
For the case of the Bloch hamiltonian, analyticity has been
proved before in cases $d=1$, $m$ arbitrary, and $d$ arbitrary, $m=1$, see
Nenciu \cite{18,19} and Helffer, Sj\"{o}strand \cite{20}. They rely
on analytical techniques. In \cite{17} topological methods are
developed.
\begin{proposition}\label{PST_2d.prop3} In case of the non-magnetic Bloch
hamiltonian let either $d\leq 3$, $m\in\mathbb{N}$ or $d\geq 4$,
$m=1$. Then there exists a collection of smooth maps
$\mathbb{R}^d\ni k\mapsto\varphi_j(k)\in L^2(M)$,
$j=n,\ldots,n+m-1$, with the following properties\smallskip\\
(i) the family $\{\varphi_j(k)\}_{j=n,\ldots,n+m-1}$ is orthonormal
and spans the range of $P^0(k)$.\smallskip\\
(ii) each map is equivariant in the sense that
\begin{equation}\label{PST_2d.2}
\varphi_j(k)=\tau(\lambda)\varphi(k+\lambda)\quad \textrm{for all
}k\in\mathbb{R}^d\,,\;\lambda\in \Lambda\,,
\end{equation}
where $\tau(\lambda)$ is multiplication by $e^{i\lambda\cdot y}$.
The same property holds for the non-magnetic periodic Dirac operator
and Pauli operator with spin-orbit coupling.
\end{proposition}
\noindent\textbf{Remark}. The proof uses the first Chern
class of the vector bundle whose fiber at $k$ is the span of the
family $\{\varphi_j(k)\}_{j=n,\ldots,n+m-1}$ i.e. $\mathrm{Ran}\,P^0(k)$. To
establish continuity, and thus smoothness, this first Chern class
has to vanish, a property, which does not hold for a magnetic Bloch
hamiltonian except for some particular energy bands.\medskip

If $\varepsilon$ is small, excitations across the energy gap are
difficult to achieve. More precisely to $P^0$ one can associate a
projection operator $\Pi^\varepsilon$ such that for arbitrary
$\ell,\ell'\in \mathbb{N}$, $\tau\in \mathbb{R}_+$, it holds
\begin{equation}\label{PST_2d.1a}
\|(1-\Pi^\varepsilon)e^{-iH^\varepsilon t}\Pi^\varepsilon\psi\|\leq
c_{\ell,\ell'}(\tau) \varepsilon^\ell\|\psi\|
\end{equation}
for $0\leq t\leq\varepsilon^{-\ell'}\tau$ with suitable constants
$c_{\ell,\ell'}(\tau)$ independent of $\varepsilon$. In other words that the
subspaces $\Pi^\varepsilon\mathcal{H}$ and $(1-\Pi^\varepsilon)\mathcal{H}$ almost
decouple, i.e.~decouple at any prescribed level of precision and over any polynomial
length of the time span under consideration. For the specific case of the Bloch
hamiltonian more quantitative details on the decoupling are provided in Section
\ref{PST_sec.3}.

If the gap condition is not satisfied, the dynamical properties are much more model
dependent. Firstly the gap condition can be violated in various ways. In our
context, since $H(k)$ has discrete spectrum, the violation comes through band
crossings. The behavior close to a band crossing has to be studied separately
\cite{23,24}. In other models the energy band sits at the bottom of the continuous
spectrum of $H(k)$ without gap \cite{25}. Then an assertion like Equation
(\ref{PST_2d.1a}) holds only under a suitable restriction to small $\ell,\ell'$,
usually $\ell,\ell'=1$ or perhaps $\ell=2$, $\ell'=1$.

The inequality (\ref{PST_2d.1a}) makes no assertion about the dynamics inside the
almost invariant subspace $\Pi^\varepsilon\mathcal{H}$. While there is a general
theory available \cite{1}, we prefer to discuss the examples separately in the
subsequent sections.

\section{Nonmagnetic Bloch hamiltonians: Peierls substitution and
geometric phase corrections}\label{PST_sec.3} \setcounter{equation}{0}

We discuss in more detail the effective dynamics for the
Schr\"{o}dinger equation with a periodic potential. For concreteness
we fix the spatial dimension to be 3. Under Zak transform the
nonmagnetic Bloch hamiltonian becomes
\begin{equation}\label{PST_3.1}
\mathcal{U}_Z\Big(\frac{1}{2}\big(-i\nabla_x -A(\varepsilon
x)\big)^2 +V_\Gamma(x)+\phi(\varepsilon
x)\Big)\mathcal{U}^{-1}_Z=H^\varepsilon_Z
\end{equation}
with
\begin{equation}\label{PST_3.2}
H^\varepsilon_Z=\frac{1}{2}\big(-i\nabla_y +k
-A(i\varepsilon\nabla^\tau_k)\big)^2
+V_\Gamma(y)+\phi(i\varepsilon\nabla^\tau_k)\,.
\end{equation}
Here $\nabla^\tau_k$ is differentation with respect to $k$ and satisfying the
$y$-dependent boundary conditions (\ref{PST_2a.3}). $H^\varepsilon_Z$ is a
self-adjoint operator on $L^2_\tau(\mathbb{R}^3,H^2(M))$, compare with
(\ref{PST_2a.5}).

In (\ref{PST_3.2}) we observe that the external potentials couple the fibers. To
emphasize this feature we think of (\ref{PST_3.2}) as being obtained through Weyl
quantization from the operator valued function
\begin{equation}\label{PST_3.3}
H_0(k,r)=\frac{1}{2}\big(-i\nabla_y +k -A(r)\big)^2
+V_\Gamma(y)+\phi(r)
\end{equation}
as defined on $(r,k)\in\mathbb{R}^6$ and acting
on $\mathcal{H}_\mathrm{f}$ with fixed domain $H^2(M)$, see \cite{21} for
details. In this form one is reminded of the Weyl quantization of
the classical hamiltonian function
$h_{\mathrm{cl}}(q,p)=\frac{1}{2}p^2+V(q)$ which yields the
semiclassical hamiltonian
\begin{equation}\label{PST_3.4}
H_\mathrm{sc}=\frac{1}{2}(-i\varepsilon\nabla_x)^2+V(x)
\end{equation}
acting in $L^2(\mathbb{R}^3)$. The analysis of (\ref{PST_3.4}) yields that on the
time-scale $\varepsilon^{-1}t$ the wave packet dynamics governed by $H_\mathrm{sc}$
well approximates the flow generated by $h_\mathrm{cl}$. In contrast, the adiabatic
analysis deals with operator valued symbols, as in (\ref{PST_3.3}), and has as a
goal to establish that the dynamics decouples into almost invariant subspaces and to
determine the approximate dynamics within each such subspace.

To be specific, let us then fix throughout one band index $n$ and let us assume that
the band energy $E_n$ is nondegenerate and satisfies the gap condition. Therefore we
know that $E_n:M^\ast\to\mathbb{R}$ is smooth and we can choose the family of Bloch
functions $\varphi_n(k)$, with $H(k)\varphi_n(k)=E_n(k)\varphi_n(k)$, such that
$\varphi_n$ depends smoothly on $k$. For each $\ell\in\mathbb{N}=\{0,1,\ldots\}$
there exists then an orthogonal projection $\Pi^\varepsilon_\ell$ on
$\mathcal{H}_\tau$ such that
\begin{equation}\label{PST_3.5}
\|[H^\varepsilon_Z\,,\,\Pi^\varepsilon_\ell]\|\leq c_\ell\varepsilon^{\ell+1}
\end{equation}
for some constants $c_\ell$. Integrating in time one concludes
that the subspaces
$\Pi^\varepsilon_\ell\mathcal{H}_\tau$ are almost invariant in the
sense that
\begin{equation}\label{PST_3.6}
\|(1-\Pi^\varepsilon_\ell)e^{-i\varepsilon^{-\ell'}tH^\varepsilon_Z}\Pi^\varepsilon_\ell\psi\|\leq
\|\psi\|(1+|t|)c_{\ell}\,\varepsilon^{\ell+1}\varepsilon^{-\ell'}
\end{equation}
for any $\ell, \ell'\in\mathbb{N}$. Note that the adiabatic time scale, order
$\varepsilon^{-\ell'}$, can have any power law increase, at the expense of choosing
the order of the projection $\Pi^\varepsilon_\ell$ sufficiently large. Only for
times of order $e^{1/\varepsilon}$ one observes transitions away from the almost
invariant subspace. The zeroth order projection is attached to the $n$-th band under
consideration, while the higher orders are successively smaller corrections to
$\Pi^\varepsilon_0$. To construct $\Pi^\varepsilon_0$ one considers the projection
onto the $n$-th band, $|\varphi_n(k)\rangle\langle \varphi_n(k)|$, as an operator
valued function with values in $B(\mathcal{H}_\mathrm{f})$. From it we obtain the
minimally substituted projection $|\varphi_n(k-A(r))\rangle\langle
\varphi_n(k-A(r))|$. Its Weyl quantization is $\varepsilon$-close to the orthogonal
projection $\Pi^\varepsilon_0$.

The second task is to determine the approximate time-evolution on
$\Pi^\varepsilon_\ell\mathcal{H}_\tau$. Since the subspace changes with
$\varepsilon$, it is more convenient to unitarily map
$\Pi^\varepsilon_\ell\mathcal{H}_\tau$ to an $\varepsilon$-independent
reference Hilbert space, which in our case is simply $L^2(M^\ast)$.
The dynamics on $L^2(M^\ast)$ is governed by an effective
hamiltonian. It is written down most easily in terms of a
hamiltonian function
$h^\varepsilon_\ell:M^\ast\times\mathbb{R}^3\to\mathbb{R}$.
$h^\varepsilon_\ell$ is a smooth function. We also may regard it as
defined on $\mathbb{R}^3\times\mathbb{R}^3$ and
$\Gamma^\ast$-periodic in the first argument. $h^\varepsilon_\ell$
admits the power series
\begin{equation}\label{PST_3.7}
h^\varepsilon_\ell=\sum^\ell_{j=0}\varepsilon^j h_j
\end{equation}
with $\varepsilon$-independent functions $h_j$. The effective
quantum hamiltonian is obtained from $h^\varepsilon_\ell$ through Weyl
quantization. The index $\ell$ regulates the time scale over which the
approximation is valid and the size of the allowed error.

In \cite{21} we provide an iterative algorithm to compute $h_j$. In
practice only $h_0$ and $h_1$ can be obtained, at best $h_2$ under
simplifying assumptions. While this may look very restrictive, it
turns out that already $h_1$ yields novel physical effects as compared to
$h_0$. Even higher order corrections seem to be less significant.

To lowest order one obtains
\begin{equation}\label{PST_3.8}
h_0(k,r)=E_n(k-A(r))+\phi(r)\,,
\end{equation}
which Weyl-quantizes to
\begin{equation}\label{PST_3.9}
    \mathcal{W}^\varepsilon
    [h_0]=E_n(k-A(i\varepsilon\nabla_k))+\phi(i\varepsilon\nabla_k)
\end{equation}
acting on $L^2(M^\ast)$, where $i\nabla_k$ is the operator of differentiation with
periodic boundary conditions. (The twisted boundary conditions appearing in
(\ref{PST_3.2}) are absorbed into the unitary map of $\Pi^\varepsilon_0\mathcal{H}$
to $L^2(M^\ast)$.) In solid state physics the Weyl quantization (\ref{PST_3.9}) is
referred to as {\it Peierls substitution}. (\ref{PST_3.8}), (\ref{PST_3.9}) have a
familiar form. The periodic potential merely changes the kinetic energy
$\frac{1}{2}k^2$ of a free particle to $E_n(k)$. The main distinctive feature is the
periodicity of the kinetic energy. For example, in presence of a linear potential
$\phi$, $\phi(x)=-E\cdot x$, an electron, initially at rest, will start to
accelerate along $E$ but then turns back because of periodicity in $k$.

To first order the effective hamiltonian reads
\begin{equation}\label{PST_3.10}
h_1(k,r)=\big(\nabla \phi(r)-\nabla E_n(\widetilde{k})\wedge
B(r)\big)\cdot\mathcal{A}_n(\widetilde{k})-B(r)\cdot
\mathcal{M}_n(\widetilde{k})\,,
\end{equation}
with the kinetic wave number $\widetilde{k}=k-A(r)$. The coefficients
$\mathcal{A}_n$ and $\mathcal{M}_n$ are the geometric
phases. They carry information on the Bloch functions $\varphi_n(k)$.
$\mathcal{A}_n$ is the Berry connection given through
\begin{equation}\label{PST_3.11}
\mathcal{A}_n(k)=i\langle \varphi_n(k)\,,\; \nabla_k
\varphi_n(k)\rangle_{\mathcal{H}_f}
\end{equation}
and $\mathcal{M}_n $ is the Rammal-Wilkinson phase given trough
\begin{equation}\label{PST_3.12}
\mathcal{M}_n(k)=\frac{1}{2}i\langle\nabla_k
\varphi_n(k),\wedge(H(k)-E_n(k))\nabla_k \varphi_n(k)\rangle_{\mathcal{H}_f}\,.
\end{equation}

The Bloch functions $\varphi_n$ are only determined up to a smooth phase
$\alpha(k)$, i.e.~instead of $\varphi_n(k)$ one might as well use
$e^{-i\alpha(k)}\varphi_n(k)$ with smooth $\alpha:M^\ast\to \mathbb{R}$.
Clearly $\mathcal{M}_n$ is independent of the gauge field $\alpha$.
On the other hand, $\mathcal{A}$ is gauge-dependent while its curl
\begin{equation}\label{PST_3.13}
\Omega_n=\nabla\wedge\mathcal{A}_n
\end{equation}
is gauge independent. From time-reversal one concludes that
\begin{equation}\label{PST_3.14}
\Omega_n(-k)=-\Omega_n(k)\,.
\end{equation}
In particular, in dimension $d=2$ for the first Chern number of the
Bloch vector bundle one obtains
\begin{equation}\label{PST_3.15}
\int_{M^\ast}dk\Omega_n(k)=0\,.
\end{equation}

For the magnetic Bloch hamiltonian, (\ref{PST_3.14}) is violated in general, see
Section \ref{PST_sec.4}. The integral in (\ref{PST_3.15}) can take only integer
values (in the appropriate units) and the first Chern number may be different from
zero. Physically this leads to the quantization of the Hall current \cite{21,22}.

We still owe the reader precise a statement on the error in the approximation. At
the moment we work in the representation space at precision level $\ell=1$. Let
$H_\mathrm{eff}$ be the Weyl quantization of $h_0+\varepsilon h_1$, see
(\ref{PST_3.8}) and (\ref{PST_3.10}). There is then a unitary
$U^\varepsilon:\Pi^\varepsilon_1\mathcal{H}_\tau\to L^2(M^\ast)$ such that for all
$\psi\in \mathcal{H_\tau}$
\begin{equation}\label{PST_3.16}
\|\big(e^{-iH^\varepsilon_Z t}-U^{\varepsilon\ast}
e^{-iH_{\mathrm{eff}}t}U^\varepsilon\big)
\Pi^\varepsilon_1\psi\|\leq c\|\psi\|(1+|\tau|)\varepsilon^2
\end{equation}
with $|t|\leq\varepsilon^{-1}\tau$ and some constant $c$ independent
of $\|\psi\|$, $\tau$, and $\varepsilon$.


\goodbreak

\section{Magnetic Bloch hamiltonians: the Hofstadter butterfly}\label{PST_sec.4}
\setcounter{equation}{0}

We turn to a magnetic Bloch hamiltonian in the form (\ref{PST_1.4}), in dimension
$d=2$ and with a transverse constant magnetic field $B_0$. We want to explain how
the limits $B_0 \to \infty$ and $B_0 \to 0$ can be understood with adiabatic
methods.  As a remark, it is worthwhile to recall that, when the physical constants
are restored, the dimensionless parameter $B_0$ is given by
\begin{equation}\label{PST_B magnitude}
B_0 = \frac{\mathcal{B}_0 S}{2 \pi \hbar c / e}\,,
\end{equation}
where $S$ is the area of the fundamental cell of $\Gamma$ and
$\mathcal{B}_0$ the strength of the magnetic field, both expressed
in their dimensional units. Thus $B_0$ corresponds physically to the
magnetic flux per unit cell divided by  $hc/e$, as the fundamental quantum of
magnetic flux. This section is based essentially on
\cite{FaurePanati}, which elaborates on previous related results
\cite{Bel1986, 20}.

Adiabatic limits are always related to separation of time-scales. In the present
case, one indeed expects that as $B_0 \to \infty$ the cyclotron motion induced by
$B_0$ is faster than the motion induced by $V_{\Gamma}$, while in the limit $B_0 \to
0$ the microscopic variations of the wave function induced by $V_{\Gamma}$ are
expected to be faster than the cyclotron motion.

Let us focus first on the Landau regime $B_0 \to \infty$. In order to make quantitative
the previous claim, one introduces the operators
\begin{equation}\label{PST_L operators}
\left\{%
\begin{array}{ll}
    L_1 = \frac{1}{\sqrt{B_0}}\(p_1 + \frac{1}{2}B_0 \, x_2  \)\,,\\
     & \qquad \qquad [L_1,L_2]=i \1, \\
    L_2 = \frac{1}{\sqrt{B_0}}\(p_2 - \frac{1}{2}B_0 \, x_1  \)\,, \\
\end{array}%
\right.
\end{equation}
and the complementary pair of operators
\begin{equation}\label{PST_G operators}
\left\{%
\begin{array}{ll}
    G_1 = \frac{1}{B_0}\(p_1 - \frac{1}{2}B_0 \, x_2  \)\,,\\
     & \qquad \qquad [G_1,G_2]= \frac{i}{B_0} \1, \\
    G_2 = \frac{1}{B_0}\(p_2 + \frac{1}{2}B_0 \, x_1  \)\,, \\
\end{array}%
\right.
\end{equation}
where the relative sign is chosen such that $[L_i, G_j]=0$, for $i,j=1,2$.

If $V_{\Gamma} =0$, then $H_{\rm MB}$ describes a harmonic
oscillator, with eigenfunctions localized on a scale $|B_0|^{-1/2}$;
this corresponds to the cyclotron motion in classical mechanics.
Since $[G_i, H_{\rm MB}] =0$, the operators $G_1$ and $G_2$ describe
conserved quantities, which semiclassically correspond  to the
coordinates of the center of the cyclotron motion.

If $V_{\Gamma} \neq 0$, but the energy scale $ \| V_{\Gamma}\|$ is
smaller than the cyclotron energy $\approx B_0$, then the operators
$G_i$ have a non-trivial but slow dynamics. By introducing the
adiabatic parameter $\s = 1/B_0$ the hamiltonian reads
\begin{equation}
H_{\rm MB} = \frac{1}{2 \s} \( L_1^2 + L_2^2 \) + V_{\Gamma}\( G_2 - \sqrt{\s}L_2,
\, - G_1 + \sqrt{\s} L_1\).
\end{equation}

In view of the commutator $[G_1,G_2] = i \s \1$, one can regard $\s
\, H_{\rm MB}$ as the $\s$-Weyl quantization (in the sense of the
mapping $(q,p)\mapsto(G_1,G_2)$) of the operator-valued symbol
\begin{equation}
    h_{\rm MB}(q,p) = \frac{1}{2} \( L_1^2 + L_2^2 \) + \s \, V_{\Gamma}\( p -
\sqrt{\s}L_2, - q + \sqrt{\s}L_1\).
\end{equation}
For each fixed $(q,p) \in \R^2$, $h_{\rm MB}(q,p)$ is an operator acting in the
Hilbert space $\Hf \cong L^2(\R)$ corresponding to the fast degrees of freedom. If
$\| V_{\Gamma}\|_{\B(\Hi)} <  \infty$, then $h_{\rm MB}(q,p)$ has purely discrete
spectrum, with eigenvalues
$$
\lambda_{n,\, \s}(q,p) = (n + \frac{1}{2}) + \s V_{\Gamma}(p,-q) +
\Or(\s^{3/2}), \qquad n \in \N,
$$
as $\s \downarrow 0$. The index $n \in \N$ labels the
\emph{Landau levels}. For $\s$ small enough, each eigenvalue band is
separated from the rest of the spectrum by an uniform gap. Thus we
can apply space-adiabatic perturbation theory to show that the band
corresponds to an almost-invariant subspace $\Pi_{n, \s}
L^2(\R^2)$. Let us focus on a specific $n \in \N$. One can prove
that the dynamics inside $\mathrm{Ran}\, \Pi_{n, \s} L^2(\R^2)$ is
described by an effective hamiltonian, which at the first order of
approximation in $\s$ reads
\begin{equation}\label{PST_Landau effective}
    h^{\s}_1 = (n + \frac{1}{2}) + \s V_{\Gamma}(G_1,-G_2).
\end{equation}

The first term in (\ref{PST_Landau effective}) is a multiple of the identity, and as
such does not contribute to the dynamics as far as the expectation values of
observables are concerned. Leading-order dynamics is thus described by the second
term, which does not depend on the Landau level $n \in \N$. Since $V_{\Gamma}$ is a
biperiodic function and $(G_1, G_2)$ a canonical pair, the second term is a
Harper-like operator.  The spectrum of such operators exhibit a complex fractal
behavior (\emph{Hofstadter butterfly}) sensitively depending on the diophantine
properties of $\alpha = \frac{B_0}{2 \pi}$ (notice that $V_{\Gamma}(G_1, G_2)$
depends on $\alpha$ through the commutator $[G_1,G_2] = i B_0^{-1} \1$). The Cantor
structure of the spectrum was  proven first in \cite{BeS} for the case
$V_{\Gamma}(x_1,x_2) = \lambda \cos x_1 + \cos x_2$ (Harper model), for a dense set
of the parameter values. Later Helffer and Sj\"{o}strand accomplished a detailed
semiclassical analysis of the Harper operator \cite{HS_Harper}. As a final step the
Cantor spectrum has been proven by Puig ($\lambda \neq 0$, $\alpha$ Diophantine)
\cite{Puig}, and by Avila and Jitomirskaya \cite{AvilaJitomirskaya} for all the
conjectured values of the parameters: $\lambda \neq 0$, $\alpha$ irrational (the
\emph{Ten Martini conjecture}, as baptized by B. Simon).

Secondly we turn to the opposite limit $B_0 \to 0$, where the slow part of the
dynamics is still described by the magnetic momentum operators $\widetilde{L}_j =
\sqrt{B_0} L_j$ $(j=1,2)$, with commutator of order  $\Or(B_0)$. However the
decomposition given by (\ref{PST_L operators}) and (\ref{PST_G operators}) is no
longer convenient.

Since $A_0$ is a linear function, $A_0(\epsi x) = \frac{1}{2} \epsi B_0 \wedge x $,
the slow variation limit $\epsi \to 0$ agrees with the weak field limit $B_0 \to 0$.
We then pose $\epsi=B_0$ and we regard $H_{\rm MB}$ in (\ref{PST_1.4}) as an
adiabatic perturbation of the periodic hamiltonian (\ref{PST_2a.1a}). Thus we are
reduced to the situation described in Section 3: to each isolated Bloch band of the
unperturbed hamiltonian there corresponds a subspace $\Pi_{n, \epsi}L^2(\R^2)$ which
is approximately invariant under the dynamics as $\epsi \downarrow 0$. The dynamics
inside this subspace is described by Peierls substitution (\ref{PST_3.9}), which now
reads
\begin{equation}\label{PST_Peierls magnetic}
    \mathcal{W}^{\epsi}[h_0] = E_n(k - \frac{1}{2} e_3 \wedge (i \epsi \nabla_k)),
\end{equation}
as an operator acting in $L^2(\T^2, dk)$. Here
$B_0=(0,0,\epsi)$ and $e_i$ is the unit vector in the
$i$-th direction.

Formula (\ref{PST_Peierls magnetic}) shows that the leading order effective
hamiltonian depends only on the operators $(K_1,K_2)=K$,
$$
K = k - \frac{1}{2} e_3 \wedge (i \epsi \nabla_k),
$$
which roughly speaking are the Fourier transform of the pair $(\widetilde{L}_1,
\widetilde{L}_2)$, and not on the complementary pair of operators. The same property
holds true for the effective hamiltonian $h^{\epsi}_\ell$, at any order of
approximation $\ell \in \N$, see \cite{FaurePanati}, with important consequences on
the splitting of magnetic subbands at small but finite $B_0$.

An operator in the form (\ref{PST_Peierls magnetic}), shortly written
$E_n(K_1,K_2)$, is \emph{isospectral} to an Harper-like operator, namely
$E_n(G_1,G_2)$ acting in $L^2(\R)$. Indeed the first numerical evidence of the
butterfly-like Cantor structure of the spectrum of Harper-like operators appeared
when Hofstadter investigated the spectrum of $\cos K_1 + \cos K_2$ as a function of
$\epsi$ \cite{Hofstadter1976}.  On the other side, an operator of the form
$E_n(K_1,K_2)$ is not \emph{unitarily equivalent} to the Harper operator
$E_n(G_1,G_2)$. The important geometric and physical consequences of this fact are
developed in \cite{FaurePanati}.

Having explained the two extreme cases, $B_0 \to 0$ and $B_0 \to \infty$, the reader
may wonder about the intermediate values of the magnetic field, $B_0 \approx 1$.  As
explained already in Section \ref{PST_sec.2c} it is convenient to introduce the
magnetic translations
$$
\mathcal{T}_{\alpha}= e^{-\frac{i}{2}\varphi(\alpha)}\, \exp(i B_0 \, \alpha \cdot
G), \qquad \alpha \in \Gamma_0,
$$
see (\ref{PST_Magnetic translations}) and (\ref{PST_2b.3bis}). If $B_0$ satisfy
Assumption~2, then $\{ \mathcal{T}_{\alpha} \}$ is a commutative group, thus leading
to the magnetic Zak transform (\ref{PST_2b.4}). $H_{\rm MB}$ is then a fibered
operator over the magnetic Bloch momentum $\kappa \in \T^2$. At each $\kappa$ the
spectrum of $H_{\rm MB}(\kappa)$ is pure point and the corresponding eigenvalues
$\mathcal{E}^{B_0}_{n}$ are the \emph{magnetic Bloch bands}.

In view of this structure, one might argue that the adiabatic perturbation of the
hamiltonian which includes, on top of the constant magnetic field $B_0$, a slowly
varying magnetic potential $A(\epsi x)$ as in (\ref{PST_1.5}) can be treated with
the methods of Section 3. There is however one crucial element missing. Indeed one
can still associate to each magnetic Bloch band $\mathcal{E}^{B_0}_{n}$, isolated
from the rest of the spectrum, an almost-invariant subspace ${\rm Ran}\,
\Pi^{B_0}_{n}$. On the other side the construction of the effective hamiltonian
relies on smoothness which may be impeded for topological reasons. Indeed the
analogue of Proposition \ref{PST_2d.prop3} is generically false for magnetic Bloch
hamiltonians, as well-understood  \cite{14,Novikov,16}. In geometric terminology
this fact is rephrased by saying that the magnetic Bloch bundle is generically
non-trivial (in technical sense).  This important fact has sometimes been
overlooked. For example, Assumption~B in \cite{DGR04} is equivalent to the
triviality of the magnetic Bloch bundle. Under this assumption
 the magnetic case is already covered by the results  in \cite{21}.
Thus the problem of adiabatic perturbation of a generic
magnetic Bloch hamiltonian appears to be an open, in our view challenging, problem
for the future.

\goodbreak


\section{Piezoelectricity}\label{PST_sec.5}
\setcounter{equation}{0}

In the year 1880 the brothers Jacques and Pierre Curie discovered
that some crystalline solids (like quartz, tourmaline, topaz,
\ldots) exhibit a macroscopic polarization if the sample is
strained.

It turns out that also this effect can be understood  in the
framework of adiabatically perturbed periodic hamiltonians, cf.\
\cite{27,28}.  The perturbation is now slowly in time,
\begin{equation}\label{PST_1.6'}
H_\mathrm{PE}(t)=-\frac{1}{2}\Delta_x+V_{\Gamma(\varepsilon
t)}(x,\varepsilon t)\,.
\end{equation}
If the potential $V_\Gamma(x,\epsi t)$ has no center of inversion,
i.e.  there is no point with respect to which the potential has space-reflection
symmetry, then  the slow variation of the periodic potential is expected to
generate a non-zero
current and can be shown to do so for particular examples \cite{avron1997}. By translation invariance this current if averaged over a unit cell is
everywhere the same and we denote the average current by $J^\epsi(t)$. For the
following discussion we assume that $V_\Gamma$ varies only for times in the finite
interval $[0,T]$.  Integrating the current per volume over the relevant time
interval yields the average polarization,
\[
\Delta {\bf P}^\epsi = \int_0^T\D t\, J^\epsi(t)\,.
\]
In this section we discuss results that relate the current $J^\epsi(t)$ directly to
the quantum mechanics of non-interacting particles governed by the hamiltonian
\eqref{PST_1.6'}, without the detour via the semiclassical model. For this we need
to solve the Schr\"odinger equation with initial state $\rho(0)=P (0)$ being the
spectral projection of $H_{\rm PE}(0)$ below the Fermi energy $E(0)$. Since the
piezo effect occurs only for insulators, we can assume that $E(0)$ lies in a gap of
the spectrum of $H_{\rm PE}(0)$ and, in order to simplify the discussion, we also
assume that this gap does not close in the course of time. Hence there is a
continuous function $E:[0,T]\to\R$ such that $E(t)$ lies in a gap of $H_{\rm PE}(t)$
for all $t$. The state at time $t$ is given by
\[
\rho^\epsi(t) = U^\epsi(t,0)\, P (0)\, U^\epsi(t,0)^*\,,
\]
where the unitary propagator $U^\epsi(t,0)$ is the solution of the
time-dependent Schr\"odinger equation
\begin{equation}
\I\epsi\frac{\D}{\D t}\,U^\epsi(t,0) = H_{\rm PE}(t) \,U^\epsi(t,0)
\qquad\mbox{with}\quad U^\epsi(0,0)= {\bf 1}\,.
\end{equation}
With the current operator  given by
\begin{equation}\label{PST_curr}
j^\epsi :=      \frac{ \I }{\epsi} \,  [ H(t), x] =
-\frac{\I}{\epsi}\nabla_x\,,
\end{equation}
and the trace per volume defined as
\begin{equation}        \label{PST_TraceperVolume}
\mathcal T (A) :=  \lim_{\Lambda_n\to\mathbb{R}^3}\frac{1}{|\Lambda_n|}
\re\, \Tr ({\bf 1}_{\Lambda_n} A)\,,
\end{equation}
 with ${\bf 1}_{\Lambda_n}$ being the characteristic
 function of a $3$-dimensional box $\Lambda_n$ with finite volume
$|\Lambda_n|$, the average current in the state $\rho^\epsi(t)$ is
\[
J^\epsi(t) = \mathcal{T}(\rho^\epsi(t)\,j^\epsi)\,.
\]
  Finally  the  average polarization is
\begin{equation}\label{PST_defp}
\Delta \p^\epsi =  \int_{0}^{T} \!\!\!\D t \,\,\mathcal T (\rho^\epsi(t)\,J^\epsi)\, ,
\end{equation}
which is the main quantity of physical interest. The given framework allows us to describe the
macroscopic polarization of a solid by a pure \emph{bulk property}, i.e.\ independently of the shape of the sample.

In the simplest but most important case (see Paragraph~(ii) in Section~1 for a
discussion of the model), the periodic potential $V_\Gamma(x,\epsi t)$   is periodic
with respect to a time-{\it independent} lattice $\Gamma$.  For this case King-Smith
and Vanderbilt  \cite{5} derived a formula for $\Delta\p$ based on linear response
theory, which turned out to make accurate predictions for the polarization of many
materials. Their by now widely applied  formula reads
\begin{equation}  \label{PST_KSV0}
\Delta \p  =  \frac{1}{(2\pi)^3} \sum_{n=0}^{N_{\rm c}} \int_{M^*} \D k \,\, \big(
\A_n(k,T) - \A_n(k,0) \big),
\end{equation}
where the sum runs over all the occupied Bloch bands and $\A_n(k,t)$ is the Berry
connection coefficient for the $n$-th Bloch band at time $t\in \R$,
\[
\A_n(k,t)= \I \langle \varphi_n(k,t) , \nabla_k \varphi_n  (k,t) \rangle_{L^2(M)}\,.
\]
 Although $\mathcal{A}_n$ depends on the choice of the Bloch function
$\varphi_n$, the average polarization  \eqref{PST_KSV0} defines a gauge invariant
quantity, i.e.\ it is independent of the choice of Bloch functions.

In \cite{27} we show that $\Delta \p^\epsi$ defined in \eqref{PST_defp} approaches
$\Delta \p$ as given by the King-Smith and Vanderbilt formula \eqref{PST_KSV0} with
errors smaller than any power of $\epsi$, whenever the latter is well defined. More
precisely we show that under suitable technical conditions on $V_\Gamma(t)$ the
average polarization  is well defined and that for any $N\in\N$
\begin{equation}\label{PST_KSV1}
\Delta \p^\epsi =  -\frac{1}{(2\pi)^d} \int_{0}^{T} \!\!\!\D t \int_{M^*} \D k \,\,
\Theta (k,t) + \mathcal{O}(\epsi^{N})\, ,
\end{equation}
where
\begin{equation}\label{PST_theta}
\Theta(k,t):=-\I \,\tr \left( P(k,t)\,[\partial_t P(k,t),\,\nabla_k P(k,t)\,]\,\right)\,,
\end{equation}
and $P(k,t)$ is the Bloch-Floquet fiber decomposition of the spectral projector
$P(t)={\bf 1}_{(-\infty,E(t)]}(H_{\rm PE}(t))$.  Whenever all Bloch bands within
Ran$P(k,t)$ are isolated, the explicit term in \eqref{PST_KSV1} agrees with
\eqref{PST_KSV0}. Note however that \eqref{PST_KSV1} is more general, since it can
be applied also to situations where band crossings occur within the set of occupied
bands.

From the point of view of adiabatic approximation, this result is
actually quite simple, since one just needs the standard
time-adiabatic theory. At time $t=0$  the state $\rho(0)$ is just
the projection $P(0)$ onto the subspace of the isolated group of
occupied bands. Since these bands remain isolated during time
evolution, this subspace is adiabatically preserved  according to
the original adiabatic theorem of Kato \cite{26}, i.e.\
\[
\rho^\epsi(t) = P(t) + \mathcal{O}(\epsi)\,,
\]
and one can compute the higher order corrections to $\rho^\epsi(t)$ using the higher
order time-adiabatic approximation due to   Nenciu \cite{Nenciu1993}.  In order to
get explicit results, one has to do the adiabatic approximation for each fixed $k\in
M^*$ separately. This is possible since $H_{\rm PE}(t,k)$ is still fibered in $k$,
due to translation invariance with respect to a time-independent lattice. However,
since we need to differentiate with respect to $k$ in order to compute the current,
as suggested by formula \eqref{PST_theta}, the expansion needs to be done uniformly
on spaces of suitable equivariant functions. This makes the proof more technical
than expected at first sight.

Alternatively one can derive also for $H_{\rm PE}(t)$   the
semiclassical equations of motion including first order corrections:
\begin{equation}
\label{PST_sceqPE} \left\{
                \begin{array}{lcl}
                    \dot q &=& \nabla_k E_n(k,t) -  \epsi\, \Theta_n(k,t), \\[2mm]

                    \dot k &=& 0\,.
                \end{array}
\right.
\end{equation}
And again averaging the velocity over the first Brillouin zone yields the correct
quantum mechanical average current that is the contribution from the $n$-th band.

Note the striking similarity between the semiclassical corrections in
\eqref{PST_theta} and the electromagnetic field. If we define the geometric vector
potential
$$
\A_n(k,t)= \I \langle \varphi_n(k,t) , \nabla_k \varphi_n  (k,t) \rangle_{L^2(M)},
$$
and the geometric scalar potential
$$
\phi_n(k,t) = - \I \langle \varphi_n(k,t) , \partial_t \varphi_n (k,t)
\rangle_{L^2(M)},
$$
in terms of the Bloch function $\varphi_n(k,t)$ of some isolated
band, then in complete analogy to the electromagnetic fields we have
\begin{equation}\label{PST_Piezocurvature2}
\Theta_n(k,t)= -\partial_t \A_n(k,t) - \nabla_k \phi_n(k,t),
\end{equation}
 and
\begin{equation}
\Omega_n(k,t) = \nabla_k \wedge \A_n(k,t)\,.
\end{equation}

Time-dependent deformations of a crystal generically also lead to a time-dependent
periodicity lattice $\Gamma(t)$, see (\ref{PST_1.6'}). This more general situation
is considered in \cite{28,31}. Now the lattice momentum $k$ is no longer a conserved
quantity and the full space-adiabatic perturbation theory is required in order to
compute the corresponding piezoelectric current. As a result an additional term
appears in the semiclassical equations of motion, reflecting the deformation of the
lattice of periodicity.

\bigskip

\noindent {\bf Acknowlegdments.} We thank Ulrich Mauthner, Max Lein, and Christof Sparber for most informative discussions. This work has been supported by the DFG Priority
Program 1095 ``Analysis, Modeling and Simulation of Multiscale Problems'' under
Sp 181/16-3.

\def\cprime{$'$}




\end{document}